\documentclass[aps,prx, onecolumn, footnote, superscriptaddress]{revtex4-2}

\usepackage{soul}
\usepackage{relsize}
\usepackage{lmodern}
\usepackage{slantsc}
\usepackage{graphicx}
\usepackage{amsmath}
\usepackage{amssymb}
\usepackage{amsthm}
\usepackage{amsbsy}
\usepackage{mathrsfs}
\usepackage{varioref}
\usepackage{dsfont}
\usepackage{bm}
\usepackage{color}
\usepackage[usenames,dvipsnames]{xcolor}
\definecolor{myblue}{rgb}{0.2,0.2,0.8}
\definecolor{myblack}{rgb}{0,0,0}
\definecolor{myurl}{rgb}{0.1,0.1,0.4}
\usepackage[colorlinks=true,citecolor=myblue,linkcolor=myblack,urlcolor=myurl]{hyperref}
\usepackage{cleveref}
\usepackage[font = footnotesize, labelfont = bf, justification = RaggedRight]{caption}
\usepackage{subfigure}
\usepackage{booktabs} 
\usepackage{array} 
\usepackage{paralist} 
\usepackage{verbatim} 
\edef\restoreparindent{\parindent=\the\parindent\relax}
\usepackage[parfill]{parskip}
\restoreparindent

\usepackage{natbib}
\setul{}{20 pt}

\newtheorem{definition}{Definition}
\newcommand{\defref}[1]{Definition~\ref{#1}}
\newtheorem{lemma}{Lemma}
\newcommand{\lemref}[1]{Lemma~\ref{#1}}
\newtheorem{prop}{Proposition}
\newcommand{\propref}[1]{Proposition~\ref{#1}}

\DeclareMathOperator\supp{supp}

\newcommand{\<}{\langle}
\renewcommand{\>}{\rangle}

\newcommand{\lo}{{\mathcal{L}}}

\newcommand{\s}{{\mathcal{S}}}
\newcommand{\ee}{{\mathcal{E}}}
\newcommand{\xx}{{\mathcal{X}}}
\newcommand{\yy}{{\mathcal{Y}}}

\renewcommand{\aa}{{\mathcal{A}}}

\newcommand{\qq}{{\mathcal{Q}}}
\newcommand{\mm}{{\mathcal{M}}}

\newcommand{\ii}{{\mathcal{I}}}
\newcommand{\jj}{{\mathcal{J}}}

\newcommand{\re}{\mathds{R}}

\newcommand{\h}{{\mathcal{H}}}
\newcommand{\kk}{{\mathcal{K}}}
\newcommand{\hs}{{\mathcal{H}\sub{\s}}}
\newcommand{\ha}{{\mathcal{H}\sub{\aa}}}
\newcommand{\hsys}{{H\sub{\s}}}
\newcommand{\happ}{{H\sub{\aa}}}

\newcommand{\E}{\mathsf{E}}
\newcommand{\F}{\mathsf{F}}
\newcommand{\Z}{\mathsf{Z}}
\newcommand{\G}{\mathsf{G}}

\renewcommand{\P}{\mathsf{P}}

\renewcommand{\nat}{\mathds{N}}

\newcommand{\one}{\mathds{1}}
\newcommand{\onesys}{\mathds{1}\sub{\s}}
\newcommand{\oneapp}{\mathds{1}\sub{\aa}}

\newcommand{\zero}{\mathds{O}}
\newcommand{\imag}{\mathfrak{i}}

\newcommand{\tr}{\mathrm{tr}}
\newcommand{\tra}{\mathrm{tr}\sub{\aa}}

\newcommand{\sub}[1]{_{\!\mathsmaller{\, #1}}}
\newcommand{\eq}[1]{Eq.~\eqref{#1}}
\newcommand{\fig}[1]{Fig.~\ref{#1}}

\newcommand{\app}[1]{Appendix~(\ref{#1})}
\newcommand{\ket}[1]{|{#1}\rangle}

\newcommand{\avg}[1]{\langle {#1} \rangle} 
 
\newcommand{\lnn}[1]{\ln\left( {#1}\right)}

\begin{document}

\title{Thermodynamically free quantum measurements}
\author{M. Hamed Mohammady}
\email{mohammad.mohammady@ulb.be}
\affiliation{QuIC, \'{E}cole Polytechnique de Bruxelles, CP 165/59, Universit\'{e} Libre de Bruxelles, 1050 Brussels, Belgium}


\begin{abstract}
Thermal channels---the free processes allowed in the resource theory of quantum thermodynamics---are generalised to thermal instruments, which we interpret as implementing  thermodynamically free quantum measurements;  a Maxwellian demon using such measurements never violates the second law of thermodynamics. Further properties of thermal instruments are investigated and, in particular, it is shown that they only measure observables commuting with the Hamiltonian,  and they thermalise the measured system when performing a complete measurement, the latter of which indicates a thermodynamically induced information-disturbance trade-off. The demarcation of  measurements that are not thermodynamically free paves the way for a resource-theoretic quantification of their thermodynamic  cost.
\end{abstract}

\maketitle

\section{Introduction}
Quantifying the thermodynamic cost of quantum processes is one of the central goals of quantum thermodynamics \cite{ Huber2015, Bedingham2016a, Faist2017, Barato2017, DeChiara2018, Pearson2021, Chiribella2021a}. However, no consensus has yet been reached as to how such quantification   should be achieved, and in particular   a universally agreed upon definition for work remains elusive \cite{Allahverdyan2005,Talkner2007,Allahverdyan2014, Perarnau-Llobet2016a, Deffner2016a, Hovhannisyan2021, Beyer2021}. The resource theory of quantum thermodynamics \cite{Brandao2013, Goold2015, Lostaglio2018a} attempts to circumvent this issue  by addressing what can and cannot be done when we restrict ourselves to processes that are ``thermodynamically free''. Such processes are referred to as (trace preserving) \emph{thermal operations}, or thermal channels \cite{Horodecki2013, Navascues2015a, Brandao2015a, Perry2015, Lostaglio2016, Mazurek2018}, which are realised by an energy conserving unitary interaction with an auxiliary system initially prepared in thermal equilibrium with an external  bath. Energy conservation of the interaction implies no net-consumption of energy. Indeed, in such a case the distribution of work given by the celebrated two-point energy measurement protocol vanishes for all input states, and hence the average work  will  agree with the unmeasured work---the difference in expected energy evaluated before and after the unitary evolution---as both quantities  vanish. On the other hand,  thermality of the auxiliary system implies that it is freely available and hence its preparation will not accrue any costs.  Within this framework,  instead of directly quantifying the work cost for a given process   we may instead ask what resource states---which are athermal and hence not thermodynamically free---are required to augment thermal channels so that the desired process may, at least approximately, be achieved.    

In this paper, we  generalise the notion of a thermal channel to a \emph{thermal instrument}---a collection of thermal operations that sum to a thermal channel---which measures a \emph{thermal observable}. As with ordinary thermal channels, a thermal instrument is implemented by a  measuring apparatus with a probe that is initially prepared in thermal equilibrium with an external bath. However, \emph{both} the ``premeasurement'' interaction between system and probe, as well as the ensuing ``pointer objectification'' that completes the measurement process, are energy conserving.   We interpret thermal instruments as implementing a thermodynamically free measurement. The justification for this interpretation follows analogous lines of reasoning as that for ordinary thermal channels cited above and, \emph{a fortiori}, by the fact that a Maxwellian demon utilising such measurements never violates the second law of thermodynamics.

We investigate other properties of thermal instruments, showing that thermodynamic constraints vastly limit the types of measurements that can be performed. For example, it is shown that thermal instruments  only measure observables that commute with the Hamiltonian,  and   thermalise the measured system when performing a  \emph{complete}  measurement. The demarcation of  measurements that are not thermodynamically free is a first step towards a resource-theoretic  quantification of  the thermodynamic cost of measurements; while maintaining energy conservation for the measurement process, the cost of a measurement can be quantified by the  athermality that must be initially present in the probe \cite{Ahmadi2013b, Miyadera2015e, Miyadera2020a, Mohammady2021a}. This approach will also allow for such cost-quantification to depend only on the properties of the desired measurement, and not on the initial state of the measured system (and hence the final state of the probe), such as is the case for  approaches that rely on Landauer erasure of the probe \cite{Sagawa2009b, Jacobs2012a}, or where the  measurement process is embedded in a  thermodynamic cycle \cite{Lipka-Bartosik2018, Mohammady2019c}.

\section{Quantum Measurement}
 An observable of a quantum system with  Hilbert space $\hs$ is represented by  a  positive operator valued measure (POVM) \cite{PaulBuschMarianGrabowski1995, Busch1996, Heinosaari2011, Busch2016a}. For simplicity, we shall only consider finite-dimensional Hilbert spaces and  discrete  observables $\E:= \{\E_x: x \in \xx\}$, with the finite set of outcomes $\xx := \{x_1, \dots ,x_N\}$, where $\zero \leqslant \E_x \leqslant \onesys$ are the \emph{effects} of $\E$ which satisfy   $\sum_{x\in \xx} \E_x = \onesys$.    The probability of registering outcome $x$ when measuring observable $\E$ in the state $\rho$ is given by the Born rule as $p^\E_\rho(x) := \tr[\E_x \rho]$. We shall employ the short-hand notation $[\E, A]=\zero$ to indicate that the operator  $A$  commutes with all the effects of $\E$, and $[\E, \F]=\zero$ to indicate that all the effects of observables $\E$ and $\F$ mutually commute.   An observable $\E$ is sharp if  $\E_x \E_y = \delta_{x,y} \E_x$, i.e.,  if  $\E_x$ are mutually orthogonal projection operators. An observable that is not sharp will be called unsharp.

Every observable $\E$ is compatible with infinitely many instruments \cite{Davies1970}, which describe how the measured system is transformed. An instrument   acting in $\hs$ is a collection of operations (completely positive trace non-increasing linear maps) $\ii:= \{\ii_x : x\in \xx\}$ such that $\ii_\xx(\cdot) := \sum_{x\in \xx} \ii_x(\cdot)$ is a channel (a trace preserving operation). An instrument $\ii$ is identified with a unique observable  $\E$ via the relation  $\tr[\ii_x(\rho)] = \tr[\E_x \rho]$ for all outcomes $x$ and states $\rho$, and we shall refer to such $\ii$ as an $\E$-compatible instrument, or an $\E$-instrument for short, and to $\ii_\xx$ as the corresponding $\E$-channel.

Every instrument may be implemented by infinitely many measurement schemes \cite{Ozawa1984}. A measurement scheme  is given by the tuple $\mm:= (\ha, \xi, \ee , \Z)$, where $\ha$ is the Hilbert space for (the probe  of)  the apparatus $\aa$ and  $\xi$ is a fixed state on  $\ha$, $\ee $ is a channel acting in  $\hs\otimes \ha$ which serves to correlate   the two systems, and $\Z := \{\Z_x : x\in \xx\}$ is a   pointer observable acting in  $\ha$. For all outcomes $x$,  the operations of the instrument $\ii$ implemented by $\mm$ can be written as
\begin{align}\label{eq:instrument-dilation}
    \ii_x (\cdot) = \tra[(\one\sub{\s}\otimes \Z_x) \ee (\cdot \otimes \xi) ],
\end{align}
where  $\tra[\cdot]$   is the partial trace  over $\ha$. The channel implemented by $\mm$ is thus $\ii_\xx(\cdot)  = \tra[\ee (\cdot \otimes \xi) ]$.

\subsection{Thermodynamically free measurement schemes, thermal instruments, and thermal observables}
A measurement scheme $\mm$ can be understood in two ways. If we wish to measure some desired observable, with some specific choice of  instrument, then we may specify the elements of $\mm$ so as to achieve this. However, one may also consider the elements of $\mm$ as given, and then ask what  observable and instrument such a scheme implements. If we impose any constraints on $\mm$, it naturally follows that the class of implementable observables and instruments will be limited. We now define  thermodynamically free measurement schemes, and subsequently determine the class of observables and instruments such schemes may or may not implement.

\begin{definition}\label{defn:thermal-measurement}
A thermodynamically free measurement scheme for a system $\hs$ with Hamiltonian $\hsys$ is described by the tuple $\mm_\beta:= (\ha, \happ, \beta, \ee , \Z)$, where:

\begin{enumerate}[(i)]
    \item The probe  with  Hamiltonian $\happ$  is prepared in a Gibbs state $\xi_\beta :=e^{-\beta \happ}/\tr[e^{-\beta \happ}]$ with  inverse temperature $\beta >0$.
    
    \item The interaction channel $\ee $ is  bistochastic, i.e., $\ee $ preserves both the trace and the identity. 
    
    \item The interaction channel $\ee $ conserves the total additive Hamiltonian $H = \hsys\otimes \oneapp + \onesys \otimes \happ$.
    
    \item The pointer observable satisfies $[\Z,\happ]=\zero$.
\end{enumerate}

\end{definition}

Condition (i) is justified by the fact that a Gibbs state describes a system when it is in thermal equilibrium with a large thermal bath. Provided that the bath is given, for example if it is the ambient environment that an experimental situation happens to find itself in, then a Gibbs state is thermodynamically free as it  does not require any effort to prepare---one need only place $\aa$ in thermal contact with the bath, and wait a sufficiently long  time for it to reach thermal equilibrium. Indeed, in such a case there is no need to ``erase'' the information stored in the probe between successive measurements, and hence no Landauer erasure cost will ensue \cite{Reeb2013a, Miller2020}; the probe may simply be discarded and replaced with another Gibbs state.  Condition (ii) ensures that the entropy of the compound system cannot decrease due to the measurement interaction, so that the second law is satisfied and no hidden ``entropy sinks'' are being used \cite{Alberti1982}. Note that unitary channels are a special subclass of bistochastic channels.  On the other hand, conditions (iii) and (iv) are both justified by energy conservation, so that the compound of system to be measured and the probe of the apparatus are energetically isolated during the \emph{entire} measurement process. A channel $\Phi$   conserves the Hamiltonian $H$ if  all moments of energy are invariant under its action, i.e.,  $\tr[H^k \Phi(\varrho)] = \tr[H^k \varrho]$ for all states $\varrho$ and $k \in \nat$.  While a channel may preserve the first moment of energy while not the higher moments, since $\ee$ is bistochastic then preservation  of the first moment guarantees that all higher moments will be preserved. See \app{app:bistochastic-conservation} for a proof.  Now note that in order for the measurement process to be complete, the pointer observable must be objectified, or ``measured''. We model the objectification process in the language of instruments, and so introduce some $\Z$-compatible instrument $\jj$ acting in $\ha$ that measures the pointer observable. While we do not consider the exact form such an instrument takes---by \eq{eq:instrument-dilation} we see that it is only the pointer observable, and not the instrument that measures it, which uniquely determines the instrument acting in the system---we do demand that the  $\Z$-channel $\jj_\xx$ conserves energy. As shown in  Ref. \cite{Mohammady2021a}, this constraint demands that $\Z$ commutes with $\happ$. We refer to such commutation  as the the Yanase condition  \cite{Yanase1961, Ozawa2002, Loveridge2011} which was first introduced in the context of  the Wigner-Araki-Yanase theorem  \cite{E.Wigner1952,Busch2010,Araki1960}. 

\begin{figure}[!htb]
\begin{center}
\includegraphics[width=0.56\textwidth]{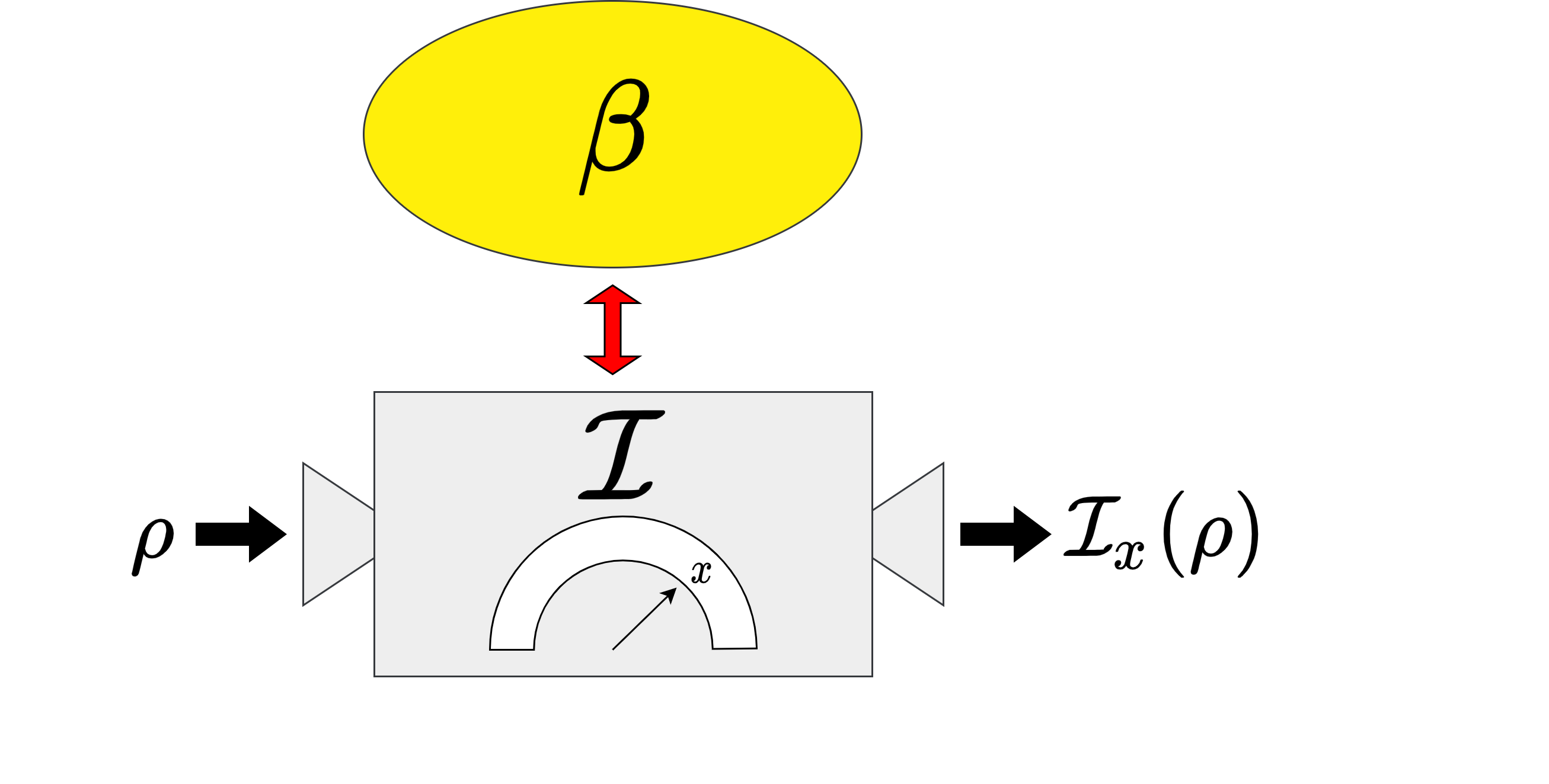}
\vspace*{-0.2cm}
\caption{An instrument $\ii$ measures the system in state $\rho$ and, conditional on observing outcome $x$, prepares the (non-normalised) state  $\ii_x(\rho)$. $\ii$ is thermal when it is implemented by a bistochastic energy conserving interaction between the  system and a probe   initially prepared in  thermal equilibrium with a bath of inverse temperature $\beta$, followed by read-out of a pointer observable commuting with the probe's Hamiltonian. } \label{fig:Instrument}
\vspace*{-0.5cm}
\end{center}
\end{figure}

\bigskip

By \eq{eq:instrument-dilation} and \defref{defn:thermal-measurement}, the operations of an instrument implemented by a thermodynamically free measurement scheme  will read $\ii_x( \cdot) = \tra[(\one\sub{\s}\otimes \Z_x) \ee (\cdot \otimes \xi_\beta) ]$. We may now define thermal instruments and thermal observables as follows:
\begin{definition}\label{defn:thermal-instrument}
Consider a system $\hs$ with Hamiltonian $\hsys$. An instrument $\ii$ acting in $\hs$   is called thermal   if there exists a thermodynamically free measurement scheme $\mm_\beta:= (\ha, \happ, \beta, \ee , \Z)$, as per \defref{defn:thermal-measurement}, such that for all $x$ and $\rho$ it holds that
\begin{align*}
\ii_x( \rho) = \tra[(\one\sub{\s}\otimes \Z_x) \ee (\rho \otimes \xi_\beta) ].
\end{align*}
 Similarly, an observable $\E$ acting in $\hs$   is called thermal   if there exists a thermodynamically free measurement scheme $\mm_\beta$ such that for all $x$ and $\rho$ it holds that
\begin{align*}
\tr[\E_x \rho]= \tr[(\one\sub{\s}\otimes \Z_x) \ee (\rho \otimes \xi_\beta) ].
\end{align*}
\end{definition}
See \fig{fig:Instrument} for a schematic representation of a thermal instrument. Note that every thermal instrument is compatible with a thermal observable, and that every thermal observable admits (possibly many) thermal instruments. But in contradistinction to the case of thermal instruments, thermality of an observable $\E$ does not impose any requirements on how  the state of the system changes upon measurement; all that is required is that the measurement statistics of the pointer observable after the interaction recovers the measurement statistics of $\E$ in the system prior to the interaction. Indeed, as we shall see below, the constraints on the thermality of observables are much weaker than the constraints on the thermality of instruments.

\subsection{Gibbs-preservation and time-translation covariance of thermal instruments, and time-translation invariance of thermal observables} 
Two salient features of thermal channels are Gibbs-preservation, and time-translation covariance. Such properties are also enjoyed by the individual operations of a thermal instrument: if $\ii$ is an $\E$-compatible thermal instrument then for all outcomes $x$ it holds that $\ii_x(\tau_\beta) = \tr[\E_x \tau_\beta]\tau_\beta$, where $\tau_\beta := e^{-\beta \hsys}/\tr[e^{-\beta \hsys}]$ is the Gibbs state of the system for some $\beta>0$, and that  $\ii_x( e^{-\imag t \hsys} \rho  \, e^{\imag t \hsys}) =  e^{-\imag t \hsys} \ii_x(\rho) e^{\imag t \hsys}$ for all $\rho$ and $t$. See \app{app:thermal-instrument-Gibbs-covariant} for the proofs.   Covariance  implies that for any asymmetry monotone $ \mathscr{I}(\hsys, \cdot)$   that obeys \emph{selective monotonicity} \cite{Zhang2017,Takagi2018}---such as the family of Wigner-Yanase-Dyson skew informations  \cite{Wigner1963, Lieb1973}, or the   quantum Fisher information \cite{PETZ2011}---and for all thermal instruments $\ii$ and states $\rho$, it holds that 
 \begin{align*}
     \mathscr{I}(\hsys, \rho) \geqslant \sum_{x\in \xx} \mathscr{I}(\hsys, \ii_x(\rho)) \geqslant \mathscr{I}(\hsys, \ii_\xx(\rho)).
 \end{align*}
 In other words, a thermal instrument will always decrease the asymmetry of a state with respect to the Hamiltonian ``on average''.  Note that even if we abandon the Yanase condition,  item (iv) in \defref{defn:thermal-measurement}, a thermal channel $\ii_\xx$ is always  covariant , in which case the relation $\mathscr{I}(\hsys, \rho) \geqslant  \mathscr{I}(\hsys, \ii_\xx(\rho))$ will continue to hold \cite{Marvian2014}. However, in such a case    covariance for the individual operations of $\ii$ will be broken, and so it may be the case that for some $\rho$ we have $\mathscr{I}(\hsys, \rho) < \sum_{x\in \xx} \mathscr{I}(\hsys, \ii_x(\rho))$.

As shown in \app{app:invariance-proof}, an observable $\E$ is thermal if and only if $\E$ is time-translation invariant, i.e., $e^{\imag t \hsys} \E_x e^{-\imag t \hsys} = \E_x $ for all $x$ and $t$, which is equivalent to $[\E, \hsys] = \zero$. The necessity of invariance follows directly from the covariance of thermal instruments. On the other hand, the sufficiency follows from the fact that we may always choose a ``trivial'' thermodynamically free measurement scheme, which uses a probe that is identical to the measured system, and a unitary swap channel which is evidently both bistochastic and energy conserving. Since any pointer observable commuting with the Hamiltonian is permitted, then by choosing $\Z = \E$, we see that all observables commuting with the Hamiltonian are thermal. However, a trivial thermodynamically free measurement scheme implements a trivial thermal instrument, i.e., the operations of the instrument will read $\ii_x(\rho) = \tr[\E_x \rho] \tau_\beta$ for all $x$ and $\rho$. That is, independent of the input and the observed outcome, the system will be thermalised; recall that the measurability question is independent of the question of how the system is transformed upon measurement. Therefore, while all  invariant observables are thermal, it does not follow that all covariant instruments are thermal. Indeed, it can be shown that some covariant instruments cannot be thermal.  As a simple example, consider the case where $\E$ is a rank-1 sharp observable, where the effects are the rank-1 projections $\E_x = |\psi_x\>\<\psi_x|$, with $\{\ket{\psi_x}\}$ an eigenbasis of $\hsys$. It is trivial to show that the von Neumann-L\"uders instrument  $\ii^L_x(\cdot) = \<\psi_x| \cdot|\psi_x\> |\psi_x\>\<\psi_x|$ is covariant. As shown in Ref. \cite{Guryanova2018}, such an instrument cannot be implemented by a rank non-decreasing interaction channel with a probe that is prepared in a full-rank state. Note that such a restriction is independent of energy conservation, and follows only from the third law of thermodynamics. But since a thermodynamically free measurement scheme employs a bistochastic interaction channel $\ee$, which is rank non-decreasing, and a thermal probe $\xi_\beta$ for the apparatus, which is full-rank, then such an instrument does not admit a thermodynamically free measurement scheme, and is hence not thermal.  In fact, as we shall show below, a thermodynamically free measurement of a rank-1 observable such as that discussed above will necessarily thermalise the measured system.

\section{Extractable work and the second law }
Given a single thermal bath of inverse temperature $\beta$, the extractable work of a system $\s$ with Hamiltonian $\hsys$, initially prepared in state $\rho$,  is defined as 
\begin{align}\label{eq:extractable-work}
W_\rho:= \beta^{-1} S(\rho \| \tau_\beta),
\end{align}
where $S(\cdot \| \cdot)$ is the quantum relative entropy. The  extractable work  is identified with the non-equilibrium free energy of $\rho$ relative to the Gibbs state $\tau_\beta$, and has an operational meaning as the   maximum amount of work that can be extracted by an isothermal process, achieved in the quasistatic limit as  $\rho$ is transformed to  $\tau_\beta$ \cite{Esposito2011}. If the system is initially prepared in thermal equilibrium with the bath then no work can be extracted, since $S(\rho \| \tau_\beta) = 0$ whenever $\rho = \tau_\beta$. This is  the second law of thermodynamics in effect. But what if we are able to measure the system? We define the average extractable work of a system initialised in state $\rho$, given  measurement by an $\E$-compatible instrument $\ii$ followed by feedback, as 
\begin{align}\label{eq:average-extractable-work}
    \avg{W_\rho^\ii} := \beta^{-1} \sum_{x\in \xx} \tr[\E_x \rho] S(\rho_x \| \tau_\beta),
\end{align}
where for any $\rho$ and $x$ such that $\tr[\E_x \rho]>0$, we define  $\rho_x := \ii_x(\rho)/\tr[\E_x \rho]$ as the  conditional post-measurement state of the system.    Here, $\avg{W_\rho^\ii}$ is the maximum (average) amount of work that can be extracted if, conditional on observing outcome $x$, we choose a specific isothermal process   so as to transform the post-measurement state $\rho_x$ to the Gibbs state $\tau_\beta$. The average extractable work therefore quantifies the merit of an information  heat engine that utilises measurement and feedback with a single thermal bath, such as the Szilard engine and its descendants \cite{Szilard1976,Maruyama2009, Kim2011a,Mohammady2017,Aydin2020}.

Now let us consider the  measurements that are usually considered in the literature. Let $\E$ be a sharp  observable with effects $\E_x = |\psi_x\>\<\psi_x|$, where $\{\ket{\psi_x}\}$ is an eigenbasis of $\hsys$, and assume that $\E$ is measured by the von Neumann-L\"uders instrument $\ii^L_x(\cdot) = \<\psi_x| \cdot|\psi_x\> |\psi_x\>\<\psi_x|$.  By a simple calculation, we can see that when the input state is thermal, i.e., $\rho = \tau_\beta$, then it holds that $\rho_x = |\psi_x\>\<\psi_x|$ for all $x$, and \eq{eq:average-extractable-work} reduces to $ \beta^{-1}\mathscr{H}$, 
where $\mathscr{H}$ is the Shannon entropy of the  probability distribution $\{\<\psi_x| \tau_\beta |\psi_x\> \}$; using measurement and feedback, we have completely converted heat into work. Previous attempts to ``save'' the second law in such a case rely on Landauer erasure of the probe, which given certain assumptions about the measurement process can be shown to  have a minimum work cost of $\beta^{-1}\mathscr{H}$. But as shown by the following proposition, maintaining the second law will not need such arguments when the measurement itself is  thermodynamically free:

\begin{prop}\label{prop:thermal-instrument-second-law}
Let $\ii$ be a thermal $\E$-instrument acting in $\hs$, implemented at inverse temperature $\beta$. Then for all states $\rho$, it holds that
\begin{align*}
 W_\rho \geqslant \beta^{-1} D_\rho^\ii + \avg{W_\rho^\ii},    
\end{align*}
with    $W_\rho$ and $\avg{W_\rho^\ii}$  defined by   \eq{eq:extractable-work} and \eq{eq:average-extractable-work}, respectively, and where $D_\rho^\ii := \sum_x \tr[\E_x \rho] \lnn{\tr[\E_x \rho]/ \tr[\E_x \tau_\beta]} \geqslant 0$ is the classical relative entropy between the probabilities arising from a measurement in $\rho$, and the probabilities arising from a measurement in $\tau_\beta$. 
\end{prop}

See \app{app:second-law} for a proof.  The above proposition is a consequence of the Gibbs-preserving property of thermal instruments, and states that when both the measuring apparatus and the measured system are in contact with a single thermal bath, then the average extractable work given  thermodynamically free measurements and feedback can never exceed the extractable work without measurement; indeed, if the system is initially prepared in thermal equilibrium with the bath, then  $\avg{W_{\tau_\beta}^\ii} = W_{\tau_\beta} =0$ and so no work can be extracted at all.  Couched in more poetic terms, a thermodynamically impotent Maxwellian demon needs no exorcism, for while it may gain information about a system in thermal equilibrium, it is unable to convert such information into useful work \cite{Norton-Demon-2}.

Now let us address the question of the heat that is absorbed by the system from the thermal environment, via the thermal probe, as a result of the measurement interaction. The channel $\Lambda(\cdot) := \tr\sub{\s}[\ee(\cdot \otimes \xi_\beta)]$, where $\tr\sub{\s}[\cdot]$ denotes the partial trace over $\hs$,  describes how the probe is transformed after it has interacted with the system in state $\rho$. $\Lambda$ is referred to as the conjugate channel to $\ii_\xx$. For any input state $\rho$ of the system, the heat absorbed by the system from the probe  is defined as the decrease in the expected energy of the probe, i.e.,
\begin{align*}
    \qq_\rho^\ii := \tr[\happ (\xi_\beta - \Lambda(\rho))].
\end{align*}
By energy conservation, it is trivial to show that $\qq_\rho^\ii = \tr[\hsys (\ii_\xx(\rho) - \rho)]$. That is, the heat absorbed by the system is identical to the increase in expected energy of the system. This is consistent with the first law of thermodynamics, and the fact that we are assuming that no work is done as a result of the measurement interaction.  Now note that we may  write $\beta^{-1}(S(\rho \| \tau_\beta) - S(\sigma \| \tau_\beta)) = \tr[\hsys (\rho - \sigma)] + \beta^{-1}(S(\sigma) - S(\rho))$ for all states $\rho, \sigma$, where $S(\cdot)$ is the von Neumann entropy \cite{Reeb2013a}. By \eq{eq:extractable-work}, \eq{eq:average-extractable-work}, and \propref{prop:thermal-instrument-second-law}, we may therefore write  the following:
\begin{align}\label{eq:heat-Groenewold}
  \avg{W_\rho^\ii} -  W_\rho &=  \qq_\rho^\ii +  \beta^{-1} I(\ii, \rho) \leqslant - \beta^{-1} D_\rho^\ii \leqslant 0,
\end{align}
where 
\begin{align*}
I(\ii, \rho) :=     S(\rho) - \sum_x \tr[\E_x \rho] S(\rho_x)
\end{align*}
is the Groenewold information gain as a system in state $\rho$ is measured by an instrument $\ii$  \cite{Groenewold1971, Lindblad1972,Ozawa1986}. The Groenewold information gain is guaranteed to be non-negative for all $\rho$ if and only if the instrument $\ii$ is ``quasi-complete'', where $\ii$ is called quasi-complete if for every pure state $\rho$, the conditional post-measurement states $\rho_x$ are also pure \cite{Ozawa1986}. For example, the L\"uders instrument  $\ii^L_x(\cdot) = \sqrt{\E_x} \cdot \sqrt{\E_x}$ compatible with an arbitrary observable $\E$ is quasi-complete. \eq{eq:heat-Groenewold}  demonstrates that a thermal instrument  cannot be quasi-complete. To see this, let us choose  $\rho = |\psi_0\>\<\psi_0|$ as a ground-state of $\hsys$ (which may have a  degenerate spectrum) so that $\qq_\rho^\ii\geqslant 0$. Now note that unless the observable measured by $\ii$ is trivial, i.e., if for all outcomes $x$ either $\E_x \propto \onesys$ or $\E_x = \zero$, then  $\tr[\E_x \rho] \ne \tr[\E_x \tau_\beta]$ for at least some $x$, which implies that $D_\rho^\ii >0$. In such a case,   the final inequality in \eq{eq:heat-Groenewold} becomes strict, i.e., $\qq_\rho^\ii +  \beta^{-1} I(\ii, \rho) < 0$, and so   $I(\ii, \rho)$ must be strictly  negative;  while $\rho$ is pure, then at least some $\rho_x$ are mixed.  In general,   the Groenewold information gain for a thermal instrument will be negative, and the more negative it is, the smaller the extractable work  from measurement and feedback becomes in comparison to the extractable work without measurement.

Let us now also highlight a simple consequence of \eq{eq:heat-Groenewold}, which is that for any thermal instrument implemented at inverse temperature $\beta$, the heat absorbed from the probe will obey the bound
\begin{align*}
 \qq_\rho^\ii \leqslant    -\beta^{-1} I(\ii, \rho).
\end{align*}
The same inequality was shown to hold in Theorem 1 of Ref. \cite{Danageozian2022} in a different setting. The authors in Ref. \cite{Danageozian2022} also assumed that the measuring apparatus is prepared in a Gibbs state, but they did not impose energy conservation on the measurement interaction, and instead considered the class of all unitary interaction channels.  Moreover, the  inequality was shown to hold only in the case where $\Lambda(\rho) = \xi_\beta$, i.e., where the reduced state of the probe does not change as a result of the measurement interaction. Such an approximation was argued to be justified if the probe is ``macroscopic''. But in the case of thermal instruments, the above inequality holds irrespective of how the state of the probe changes, and also when the energy conserving interaction channel is not unitary but is bistochastic. 

\section{Energy compatibility }
Recall that an observable $\E$ is thermal if and only if $[\E, \hsys]=\zero$.  Such commutativity  admits an elegant interpretation in terms of  \emph{compatibility}  \cite{Heinosaari2015,Guhne2021}. Two observables $\E:= \{\E_x : x\in \xx\}$ and $\F:= \{\F_y: y \in \yy\}$ are compatible, or jointly measurable, if they admit a joint observable $\G:=\{\G_{x,y}: (x,y) \in \xx\times \yy\}$ so that $\E_x = \sum_y \G_{x,y}$ and $\F_y = \sum_x \G_{x,y}$. If $\E$ and $\F$ do not admit a joint observable, then they are \emph{incompatible}. Now let the Hamiltonian have the spectral decomposition $\hsys = \sum_m \omega_m \P_m$, with $\omega_m$ the distinct energy eigenvalues and $\P_m$ the corresponding spectral projections. The sharp observable $\P:= \{\P_m\}$ is the spectral measure of $\hsys$, and we will refer to measurement of $\P$ and of $\hsys$ interchangeably. Given that $[\E, \hsys]=\zero$ implies  $[\E, \P]=\zero$, and that  commutativity is a sufficient condition for compatibility, then $\E$ and $\P$ are jointly measureable. Indeed, since $\P$ is sharp then the effects of the joint observable are uniquely given as $\G_{x,m} = \E_x \P_m$. In other words, for any thermal observable $\E$ we may  construct a single measurement device that jointly gives both the statistics of $\E$ and the statistics of the Hamiltonian.

 Of course, while all thermal observables are  jointly measureable with the Hamiltonian, this does not generally imply that  two thermal observables are themselves compatible; while any pair of thermal observables $\E$ and $\F$ must both commute with $\hsys$, it may be possible for them to not commute with each other. However, there is one limiting situation where compatibility is guaranteed: when the Hamiltonian has a non-degenerate energy spectrum. In such a case $\P_m$ are rank-1 projections, and so  the effects of any thermal observable are simultaneously diagonalisable as $ \E_x = \sum_m p(x|m) \P_m$, where $\{p(x|m)\}$ is a family of non-negative numbers that satisfy $\sum_x p(x|m) = 1$ for all $m$ \cite{Pellonpaa2014a}. It is clear that  any pair of thermal observables will commute, and are therefore compatible. In such a case, we may infer that measurement of incompatible observables will always have some thermodynamic cost---even if $\E$ is a thermal observable, $\F$ will be incompatible with $\E$ only if it is non-thermal. We note that the ability to measure incompatible observables is a crucial ingredient in many quantum phenomena such as violation of Bell inequalities \cite{Fine1982} and quantum steering \cite{Cavalcanti2017}, and incompatible observables have been shown to outperform compatible ones for quantum state discrimination \cite{Carmeli2019c}. Indeed, measurement of  incompatible observables has also been suggested as a method of  efficiently fueling quantum heat engines  \cite{Manikandan2021}.

\section{Conditional state preparation and complete  measurements}
An $\E$-compatible instrument $\ii$ is said to be \emph{nuclear} if its operations satisfy 
\begin{align*}
    \ii_x(\rho) = \tr[\E_x \rho] \sigma_x
\end{align*}
for all $\rho$ and $x$, where $\{\sigma_x\}$ is a family of states that depend only on the measurement outcome, and not the input state $\rho$. Nuclear instruments have utility as conditional state preparation devices, since for any input state, conditional on observing outcome $x$ we know that the system has been prepared in the state $\sigma_x$. An example of a nuclear instrument is the well-known von Neumann-L\"uders measurement, which ``collapses'' the measured system into the eigenstates of the measured observable. The following proposition, which is a consequence of the Gibbs-preserving property of thermal instruments,  implies that   a  non-trivial state preparation device always has some thermodynamic cost:   
\begin{prop}\label{prop:nuclear-thermal-instrument}
Let $\ii$ be a  thermal $\E$-instrument acting in  $\hs$. If $\ii$ is nuclear, then $\ii$ is a trivial, thermalising instrument, with operations satisfying
\begin{align*}
\ii_x(\rho) = \tr[\E_x \rho] \tau_\beta    
\end{align*}
for all $\rho$ and $x$, where $\tau_\beta$ is the Gibbs state of the system for some $\beta >0$. 
\end{prop}
See \app{app:nuclear-thermal} for a proof. Let us now highlight an important consequence of the above result for the class of rank-1 observables. An observable $\E$ is called rank-1 if all the effects are of the form $\E_x = \lambda_x P_x$, where $\lambda_x \in (0,1]$ and $P_x$ is a rank-1 projection operator. Rank-1 observables are \emph{complete} measurements,  since any observable  can be maximally  ``refined'' into a rank-1 observable \cite{Holland1990, Buscemi2005, Pellonpaa2014}; if the effects of some observable $\F$ can be diagonalised as $\F_y = \sum_i \lambda_i^{(y)} P_i^{(y)}$, then  a rank-1 observable $\E$ with effects $\E_{x} \equiv \E_{y,i} = \lambda_i^{(y)} P_i^{(y)}$ is a maximal refinement of $\F$. As shown in Corollary 1 of Ref. \cite{Heinosaari2010} (also see Theorem 2 of Ref. \cite{Pellonpaa2013a}), all  instruments compatible with a rank-1 observable are nuclear. In conjunction with   \propref{prop:nuclear-thermal-instrument}, it follows that a  thermodynamically free  measurement of a rank-1 observable   necessarily thermalises the measured system. 

As an interesting remark, let us consider again the situation where the system's Hamiltonian has a non-degenerate spectrum, so that the effects of any thermal observable $\E$ may be written as $ \E_x = \sum_m p(x|m) \P_m$, where  $\P_m$ are the rank-1 spectral projections of the Hamiltonian. It clearly follows that measuring any thermal observable other than the Hamiltonian is superfluous; one may  reconstruct the statistics of  all thermal observables by post-processing the measurement statistics of the Hamiltonian with the  numbers $p(x|m)$.    But a thermodynamically free measurement of the Hamiltonian---which is a  rank-1 observable---necessarily  thermalises the measured system. We see that there is a   thermodynamically induced information-disturbance trade-off, where by obtaining all the information that is available without expending any thermodynamic resources, we must thermalise the system so as to destroy all the information contained therein. This is analogous to the case where, in the absence of any thermodynamic constraints, measurement of an informationally complete observable completely destroys all the information in  the measured system  \cite{Hamamura-Miyadera}.

\section{ Conclusions }
By taking inspiration from the resource-theoretic approach to quantum thermodynamics,  thermodynamically free measurements have been defined as a thermal instrument, where each step of the measurement process has zero associated costs. Indeed, such measurements never lead to an advantage in work extraction from a single thermal bath, and so the Maxwell demon paradox is resolved from the outset without need for any post-hoc exorcisms.   Having provided a preliminary demarcation of non-free measurements,  it is now possible to provide a resource-theoretic quantification for the fundamental cost of  measurements. For example, by maintaining all the elements of a thermodynamically free measurement scheme except for the thermality of the probe---maintaining the bistochasticity and energy conservation of the measurement interaction, and commutation of the pointer observable with the Hamiltonian---we may continue to  avoid the conceptual difficulties that arise when trying to directly quantify the work cost of channels. In such a case, we may obtain bounds for the necessary  athermality in the initial probe preparation so as to approximately achieve the desired measurement. Alternatively, we may continue to use thermal probes for the apparatus, but augment the thermodynamically free measurement scheme by introducing extra auxiliary systems such as catalysts, so that the cost quantification would be determined by the athermality required of such  systems \cite{Lipka-Bartosik2020a, Wilming2021, Lipka-Bartosik2022}.  Insofar as the cost of measuring non-thermal observables is concerned,  i.e., observables not commuting with the Hamiltonian, a partial answer to this question has already been given in \cite{Mohammady2021a}; a large energy coherence in the probe preparation (or the catalysts), quantified by the quantum Fisher information \cite{Marvian2021}, is  necessary  to approximately measure  non-thermal observables. However, the cost of implementing non-thermal instruments in general is still largely an open problem, and this task is left for future work.

\begin{acknowledgments}
 The author wishes to thank Nicolas Cerf, Ravi Kunjwal, Harry J. D. Miller,  Ognyan Oreshkov, Patryk Lipka-Bartosik, and  M\'ario Ziman for insightful discussions. This project has received funding from the European Union’s Horizon 2020 research and innovation programme  under the Marie Skłodowska-Curie grant agreement No. 801505.
\end{acknowledgments}

\medskip 

\noindent{\bf \large Appendix}
\appendix


Before presenting the detailed proofs for the claims made in the main text, let us first
introduce some notation and basic definitions. We denote by $\lo(\h)$ the algebra of linear 
operators on a finite-dimensional complex Hilbert space $\h$, with $\zero$ and $\one$ the null 
and identity operators of $\lo(\h)$, respectively. A ``Schr\"odinger picture''  operation is 
defined as a completely positive, trace non-increasing linear map  $\Phi : \lo(\h) \to 
\lo(\kk)$, where $\h$ is the input space and $\kk$ is a potentially different output space. 
When both input and output spaces are the same, i.e.,  $\h = \kk$, we say that $\Phi$ acts in 
$\h$. The associated  ``Heisenberg picture'' dual operation is  a completely positive  linear 
map $\Phi^* : \lo(\kk) \to \lo(\h)$, defined by the trace duality $\tr[A \Phi(B)] = 
\tr[\Phi^*(A) B]$ for all $A\in \lo(\kk), B \in \lo(\h)$. $\Phi^*$ is sub-unital, i.e., 
$\Phi^*(\one\sub{\kk}) \leqslant \one\sub{\h}$, and is unital when the equality holds, which is the case exactly
when $\Phi$ is a channel, i.e., when $\Phi$ preserves the trace. A channel  
$\Phi$ acting in $\h$ is  bistochastic  if both $\Phi$ and $\Phi^*$ are trace preserving and unital. Bistochastic 
channels never decrease the von Neumann entropy of a system, since by the data processing inequality we may write 
for any state $\rho$ the following:
\begin{align*}
    S(\Phi(\rho)) - S(\rho) = S(\rho\| \one) - S(\Phi(\rho) \| \Phi(\one)) \geqslant 0,
\end{align*}
where $S(\rho):= -\tr[\rho \lnn{\rho}]$ is the von Neumann entropy of $\rho$,  and $S(\rho \| \sigma) := \tr[\rho(\lnn{\rho} - \lnn{\sigma})]$ is the entropy of $\rho$ relative to $\sigma$ whenever $\supp(\rho) \subseteq \supp(\sigma)$ and is defined as $S(\rho \| \sigma):= \infty$ otherwise,  and it holds that $S(\rho) = - S(\rho \| \one)$ \cite{Buscemi2016}.

\section{Bistochastic channels and energy conservation}\label{app:bistochastic-conservation}

A channel $\Phi$ acting in $\h$ conserves energy if  all moments of energy are preserved under its action, i.e.,  $\tr[H^k \Phi(\varrho)] = \tr[H^k \varrho]$ for all $k \in \nat$ and states $\varrho$ on $\h$. This condition can be equivalently  stated in the Heisenberg picture as $\Phi^*(H^k) = H^k$ for all $k\in \nat$. While a general channel may conserve the first moment while not conserving the higher moments, we shall now show that in the special case of bistochastic channels acting in a finite-dimensional Hilbert space, a channel conserving the first moment is guaranteed to conserve all higher moments.

\begin{lemma}\label{lemma:bistochastic-fixed-point}
Let $\Phi$ be a bistochastic channel acting in a finite-dimensional Hilbert space $\h$. Assume that $\Phi^*(H) = H$  for some self-adjoint operator with spectral decomposition $H= \sum_n \lambda_n P_n$. The following hold:
\begin{enumerate}[(i)]
    \item $\Phi^*(H^k) = H^k$ for all $k \in \nat$.
    \item $\Phi(e^{-\imag t H}  \varrho \, e^{\imag t H}) = e^{-\imag t H} \Phi(\varrho) e^{\imag t H}$ for all $\varrho \in \lo(\h)$ and $t \in \re$.
    \item $\Phi(\varrho) = \varrho$ for all states with spectral decomposition $\varrho= \sum_n p_n P_n$.
\end{enumerate}
\end{lemma}
\begin{proof}
Let   $\{K_i\}$ be any  Kraus representation for $\Phi$, so that  $\Phi(\cdot) = \sum_i K_i \cdot K_i^\dagger$ and $\Phi^*(\cdot) = \sum_i K_i^\dagger \cdot K_i$. If $\Phi$ is bistochastic, then it holds that both $\Phi$ and $\Phi^*$ are trace preserving and unital, and so    $\sum_i K_i^\dagger K_i = \sum_i K_i K_i^\dagger = \one$. Now, since $\h$ is finite-dimensional,  it follows from Theorem 3.5 of Ref. \cite{Arias2002} that the fixed-point set of both $\Phi$ and $\Phi^*$ is the commutant of $\{K_i, K_i^\dagger\}$, i.e., $\Phi(A) = A$ if and only if $[A, K_i] = [A, K_i^\dagger] = \zero$ for all $i$, and similarly $\Phi^*(A) = A$ if and only if $[A, K_i] = [A, K_i^\dagger] = \zero$ for all $i$. 

Now let us prove (i). Assume that $\Phi^*(H) = H$. By the above, it trivially follows that $\Phi^*(H^k) = \sum_i K_i^\dagger H^k K_i = \sum_i K_i^\dagger K_i H^k  = H^k$ for all $k$. Now let us note that  $H$ commutes with $K_i, K_i^\dagger$ if and only if    all spectral projections $P_n$ commute with $K_i, K_i^\dagger$. Since $e^{-\imag t H} = \sum_n e^{-\imag t \lambda_n} P_n$, (ii) follows trivially from above. Similarly for (iii), commutation of $K_i, K_i^\dagger$ with $P_n$ and the fact that $\Phi$ is bistochastic implies that $\Phi(\varrho) = \sum_i K_i \varrho K_i^\dagger = \sum_i K_i K_i^\dagger \varrho =\varrho$. 
\end{proof}

\section{Gibbs-preservation and time-translation covariance of thermal instruments}\label{app:thermal-instrument-Gibbs-covariant}

It is well-known that  thermal channels preserves the Gibbs state, and are time-translation covariant. Here, we shall show that these proprieties are also enjoyed by all operations of a thermal instrument.

We first show the Gibbs-preserving property.
\begin{lemma}\label{lemma:thermal-instrument-Gibbs}
Let  $\ii$ be a thermal $\E$-instrument acting in a system $\hs$ with Hamiltonian $\hsys$. For some $\beta>0$ there exists a Gibbs state of the system $\tau_\beta$ such that for all $x$ the following holds:
\begin{align}\label{eq:thermal-instrument-Gibbs-input}
    \ii_x(\tau_\beta) = \tr[\E_x \tau_\beta] \tau_\beta.
\end{align}
\end{lemma}
\begin{proof}
Assume that $\ii$ is a thermal instrument, so that by \defref{defn:thermal-instrument}  it admits a thermodynamically free measurement scheme $\mm_\beta:= (\ha, \happ, \beta, \ee , \Z)$ where the probe is prepared in the Gibbs state $\xi_\beta = e^{-\beta \happ}/\tr[e^{-\beta \happ}]$.  For the Gibbs state of the system with the same temperature as the probe,  $\tau_\beta = e^{-\beta \hsys}/\tr[e^{-\beta \hsys}]$,  additivity of the total Hamiltonian $H = \hsys \otimes \oneapp + \onesys \otimes \happ$ implies that $\tau_\beta \otimes \xi_\beta = e^{-\beta H}/\tr[e^{-\beta H}]$ is the Gibbs state for the total system. Since $\ee$ is bistochastic and energy conserving, it follows from item (iii) of \lemref{lemma:bistochastic-fixed-point} that $\ee(\tau_\beta \otimes \xi_\beta) = \tau_\beta \otimes \xi_\beta$, and so by \eq{eq:instrument-dilation} it holds that 
\begin{align*}
    \ii_x(\tau_\beta) &= \tr\sub{\aa}[(\onesys \otimes \Z_x) \ee (\tau_\beta \otimes \xi_\beta )] = \tr\sub{\aa}[(\onesys \otimes \Z_x) \tau_\beta \otimes \xi_\beta] = \tr[\Z_x \xi_\beta] \tau_\beta 
\end{align*}
for all $x$. But since $\ii$ is compatible with $\E$, then it must hold that $\tr[\ii_x(\tau_\beta)] = \tr[\E_x \tau_\beta]$. This implies that   $\tr[\Z_x \xi_\beta] = \tr[\E_x \tau_\beta]$. This completes the proof.
\end{proof}

Now we shall show the covariance property. We note that this result is a consequence of  Theorem 8 in Ref. \cite{Keyl1999}. 

\begin{lemma}\label{lemma:thermal-instrument-covariance}
Let $\ii$ be a thermal instrument acting in a system $\hs$ with Hamiltonian $\hsys$.  It follows that $\ii$ is time-translation covariant.
\end{lemma}
\begin{proof}
By \defref{defn:thermal-instrument}, if $\ii$ is a thermal instrument then it admits a thermodynamically free measurement scheme $\mm_\beta:= (\ha, \happ, \beta, \ee , \Z)$. By additivity of the total Hamiltonian $H = \hsys \otimes \oneapp + \onesys \otimes \happ$,  the unitary representation of the time-translation symmetry group $\re$ in $\hs \otimes \ha$ factorises as $ e^{-\imag t H} = e^{-\imag t \hsys} \otimes  e^{-\imag t \happ}$. Since $\ee$ is bistochastic and conserves the Hamiltonian, by item (ii) of \lemref{lemma:bistochastic-fixed-point} it holds that $\ee$ is time-translation covariant, i.e., $\ee (e^{-\imag t H} \varrho \, e^{\imag t H}) = e^{-\imag t H} \ee (\varrho) e^{\imag t H}$  holds for all $\varrho$ and $t$.  From \eq{eq:instrument-dilation}, we thus have for all $\rho$, $x$ and $t$ the following:
\begin{align*}
\ii_x(e^{-\imag t \hsys} \rho \,  e^{\imag t \hsys}) & = \tra[(\one\sub{\s} \otimes \Z_x ) \ee (e^{-\imag t \hsys}  \rho \, e^{\imag t \hsys} \otimes \xi_\beta) ] \nonumber \\
&  = \tra[(\one\sub{\s} \otimes \Z_x ) \ee ( e^{-\imag t H} (  \rho  \otimes \xi_\beta ) e^{\imag t H} )] \nonumber \\
&  = \tra[(\one\sub{\s} \otimes \Z_x )  e^{-\imag t H} \ee  (  \rho  \otimes \xi_\beta ) e^{\imag t H}] \nonumber \\
&  = e^{-\imag t \hsys} \tra[(\one\sub{\s} \otimes e^{\imag t \happ} \Z_x e^{-\imag t \happ} )   \ee  ( \rho  \otimes \xi_\beta )] e^{\imag t \hsys} \nonumber \\
&  = e^{-\imag t \hsys} \ii_x(  \rho ) e^{\imag t \hsys}.
\end{align*}
As such, $\ii$ is covariant. In the second line, we have used the fact that since $\xi_\beta$ is a Gibbs state then $[\xi_\beta, \happ]=\zero$, which implies  that $\xi_\beta = e^{- \imag t \happ} \xi_\beta \,  e^{\imag t \happ}$. In the third line, we have used time-translation covariance of $\ee$. In the fourth line, we have used the property of the partial trace. In the final line, we have used the Yanase condition $[\Z, \happ]=\zero$ which implies that $e^{\imag t \happ} \Z_x e^{-\imag t \happ} = \Z_x$, together with \eq{eq:instrument-dilation}. 
\end{proof}

\section{Time-translation  invariance of thermal observables}\label{app:invariance-proof}

Here we shall show that an observable is thermal if and only if it commutes with the Hamiltonian. Note, however, that not all time-translation covariant instruments are thermal. 

\begin{lemma}\label{lemma:thermal-observable-invariance}
Let  $\E$ be an observable acting in a system $\hs$ with Hamiltonian $\hsys$.  $\E$ is a thermal observable if and only if $\E$ commutes with $\hsys$. 
\end{lemma}
\begin{proof}
Let us first show the only if statement. Recall that an instrument $\ii$ is compatible with observable $\E$ if it holds that $\tr[\ii_x(\rho)] = \tr[\E_x \rho]$ for all $\rho$ and $x$, which can equivalently be stated as   $\E_x = \ii_x^*(\onesys)$ for all $x$. Now assume that $\E$ is a thermal observable, so that by \defref{defn:thermal-instrument} it admits a thermodynamically free measurement scheme, and hence must be compatible with a thermal instrument. By \lemref{lemma:thermal-instrument-covariance}, all thermal instruments are covariant. Since covariance in the Schr\"odinger picture trivially implies covariance in the Heisenberg picture, this implies that 
\begin{align*}
    \E_x &=  \ii^*_x(e^{\imag t \hsys} \onesys e^{-\imag t \hsys})= e^{\imag t \hsys} \ii^*_x(\onesys)e^{-\imag t \hsys} = e^{\imag t \hsys} \E_x e^{-\imag t \hsys} =: \E_x(t)
\end{align*}
holds for all $x$ and $t$. That is, covariance of $\ii$ implies invariance of $\E$. While  commutation of $\E$ with $\hsys$ trivially implies that $\E_x(t) = \E_x$, we now show that the converse implication also holds.    Given that $\frac{d}{dt} \E_x(t) = -\imag [\E_x(t), \hsys]$, we may write
 \begin{align*}
     \E_x(t) = \E_x -\imag \int_0^t d t_1 [\E_x(t_1), \hsys]. 
 \end{align*}
  But since $\E_x(t) = \E_x$  for all $t$, the above equation simplifies to $\E_x = \E_x - \imag t [\E_x, \hsys]$ for all $t$, which  can only be satisfied if it holds that $[\E_x, \hsys]=\zero$. 

Now we prove the if statement. Assume that $\E$ commutes with $\hsys$. By \defref{defn:thermal-instrument}, $\E$ is a thermal observable if it admits a thermodynamically free measurement scheme $\mm_\beta:= (\ha, \happ, \beta, \ee , \Z)$. Let us choose $\mm_\beta$ to be ``trivial'', i.e., let us choose a probe  that is identical to the measured system,    $\ha \simeq \hs$ and $\happ = \hsys$, which implies that the Gibbs states for the systems are also identical,   i.e., $\xi_\beta = \tau_\beta$. Let us also choose $\ee$ as a unitary swap channel acting in $\hs\otimes \ha$, so that $\ee(A\otimes B) = B \otimes A$ for all $A, B$. Such an $\ee$ is clearly bistochastic and conserves the total Hamiltonian. Finally, since $[\E,\hsys]=\zero$, then we may choose $\Z = \E$, which satisfies the Yanase condition. By \eq{eq:instrument-dilation}, the operations of the implemented instrument $\ii$ read
\begin{align*}
 \ii_x(\rho) &= \tr\sub{\aa}[(\onesys \otimes \E_x) \ee(\rho \otimes \tau_\beta)]  = \tr\sub{\aa}[(\onesys \otimes \E_x) \tau_\beta \otimes \rho] = \tr[\E_x \rho]\tau_\beta  
\end{align*}
for all $x$ and $\rho$. Since $\tr[\ii_x(\rho)] = \tr[\E_x \rho]$ for all $x$ and $\rho$, then the measured observable is $\E$, and so any $\E$ commuting with $\hsys$ is a thermal observable. 
\end{proof}
Let us highlight the fact that the proof for sufficiency of $\E$ commuting with $\hsys$ used a trivial  measurement scheme, which implements a  trivial instrument, that is,  for all input states $\rho$ and outcomes $x$, the system will be transformed to the Gibbs state $\tau_\beta$.  Therefore, while every observable commuting with the Hamiltonian is thermal, it may be the case that not all time-translation covariant instruments are thermal. Specifically, it is possible that some thermal observables do not admit a thermal instrument that will not thermalise the system. In fact, this is precisely the case for rank-1 observables. 

\section{Proof of Proposition 1}\label{app:second-law}

Let $\mm_\beta := (\ha, \happ, \beta, \ee, \Z)$ be a thermodynamically free measurement scheme for the observable $\E$ with instrument $\ii$ acting in $\hs$. Now, let us  define the instrument $\Phi$ with operations $\Phi_x : \lo(\hs) \to \lo(\hs \otimes \kk), \rho \mapsto \ii_x(\rho) \otimes |x\>\<x|$ where $\{\ket{x}\}$ is an orthonormal basis that spans $\kk$. Note that $\Phi$ is an entirely fictitious dilation, and is employed only to facilitate the proof, and should not be assigned any physical interpretation.  The action of the channel $\Phi_\xx(\cdot) := \sum_{x\in \xx} \Phi_x(\cdot)$ on the states $\rho$ and $\tau_\beta$ can be written as
\begin{align*}
    \Phi_\xx(\rho) &= \sum_{x\in \xx} \ii_x(\rho) \otimes |x\>\<x| = \sum_{x\in \xx} \tr[\E_x \rho] \rho_x \otimes |x\>\<x|, \nonumber \\
     \Phi_\xx(\tau_\beta) &= \sum_{x\in \xx} \ii_x(\tau_\beta) \otimes |x\>\<x| = \sum_{x\in \xx} \tr[\E_x \tau_\beta] \tau_\beta \otimes |x\>\<x|,
\end{align*}
where  we define $\rho_x := \ii_x(\rho)/\tr[\E_x \rho]$ for any $\rho$ and $x$ such that $\tr[\E_x \rho]>0$, and $\rho_x := \zero$ otherwise, and where the second line follows from \lemref{lemma:thermal-instrument-Gibbs}. Now define by $\bm{p}:= \{p(x) = \tr[\E_x \rho] : x\in \xx\}$ and $\bm{q}:= \{q(x) =\tr[\E_x \tau_\beta] : x\in \xx\}$ the probability vectors arising from a measurement of $\E$ in the states $\rho$ and $\tau_\beta$, respectively. By  the data processing inequality, and the ``direct sum'' property of the relative entropy (Proposition 4.3 of Ref. \cite{Khatri2020}), it holds that
\begin{align*}
    S(\rho \| \tau_\beta) &\geqslant  S(\Phi_\xx(\rho) \| \Phi_\xx(\tau_\beta))= D(\bm{p}\|\bm{q} ) +  \sum_{x\in \xx} \tr[\E_x \rho]  S(\rho_x \| \tau_\beta),
\end{align*}
where $S(\rho \| \sigma) := \tr[\rho (\lnn{\rho} - \lnn{\sigma})] \geqslant 0$ is the quantum relative entropy between states $\rho$ and $\sigma$ whenever $\supp(\rho) \subseteq \supp(\sigma)$, vanishing if and only if $\rho = \sigma$,  and $D(\bm{p}\|\bm{q} ) := \sum_x p(x) \lnn{p(x)/ q(x)} \geqslant 0$ is the classical relative entropy between probability vectors $\bm{p}$ and $\bm{q}$ whenever $p(x)>0 \implies q(x) >0 $, vanishing if and only if  $\bm{p} = \bm{q}$. Since a Gibbs state is full-rank, then  $\supp(\rho) \subseteq \supp(\tau_\beta)$ holds  for all $\rho$, while  $\tr[\E_x \tau_\beta]>0$  holds for all $x$ for which $\E_x \ne \zero$. As such,   the quantities on both sides of the above equation are always finite and non-negative.   By \eq{eq:extractable-work}, \eq{eq:average-extractable-work}, and the above, we thus obtain the bound
\begin{align*}
  W_\rho \geqslant \beta^{-1} D(\bm{p}\|\bm{q} ) + \avg{W_\rho^\ii} \geqslant \avg{W_\rho^\ii}.
\end{align*}
Since $D(\bm{p}\|\bm{q} ) > 0$ whenever the probability vectors $\bm{p}$ and $\bm{q}$ differ in at least one entry, then it will hold that $\avg{W_\rho^\ii} = W_\rho$ only if $\bm{p} = \bm{q}$.

\section{Proof of Proposition 2}\label{app:nuclear-thermal}

Recall from \lemref{lemma:thermal-instrument-Gibbs} that  if $\ii$ is a thermal $\E$-instrument, then its operations must satisfy \eq{eq:thermal-instrument-Gibbs-input}, i.e., it must hold that $\ii_x(\tau_\beta) = \tr[\E_x \tau_\beta] \tau_\beta$ for all $x$. Now, an  $\E$-compatible instrument $\ii$ is  nuclear if its operations satisfy
\begin{align}\label{eq:nuclear-instrument}
    \ii_x(\rho) = \tr[\E_x \rho] \sigma_x
\end{align}
for all $\rho$ and $x$, where $\{\sigma_x\}$ is a family of states that are independent of the input state $\rho$.  Comparing \eq{eq:thermal-instrument-Gibbs-input} with 
 \eq{eq:nuclear-instrument} for the input state $\rho = \tau_\beta$ demonstrates that if $\ii$ is a nuclear thermal  instrument, then we must have $\sigma_x = \tau_\beta$ for all $x$. As such,  the operations of  $\ii$ must satisfy 
$
\ii_x(\rho) = \tr[\E_x \rho] \tau_\beta    
$
 for all $\rho$ and $x$.

\bibliography{Projects-Thermal-Measurement.bib}

\end{document}